\documentclass[11pt]{article}

\usepackage{fullpage,amsmath,xspace,amssymb,epsfig,graphicx,cite, amssymb}
\usepackage[usenames,dvipsnames,dvips]{color}
\newtheorem{Thm}{Theorem}

\newtheorem{Cor}{Corollary}
\newtheorem{Prop}{Proposition}
\newtheorem{Claim}{Claim}

\newtheorem{Def}{Definition}

\newenvironment{proof}{\noindent {\textbf{Proof }}}{$\Box$ \medskip}

\newcommand\mbR{\mbox{$\mathbb{R}$}}

\newcommand {\ie} {\textit{i.e.}\xspace}
\newcommand {\st} {\textit{s.t.}\xspace}

\newcommand\pr{\mbox{\bf Pr}}
\newcommand\av{\mbox{\bf{\bf E}}}

\newcommand\COMMENT[1]{}

\begin{document}
\title{A quantum protocol for sampling correlated equilibria unconditionally and without a mediator}
\author{Iordanis Kerenidis \thanks{Laboratoire d'Informatique Algorithmique: Fondements et Applications, Univ Paris Diderot 7, and CNRS; Centre for Quantum Technologies, Singapore. Email:   {\tt jkeren@liafa.jussieu.fr}} \quad and \quad Shengyu Zhang\thanks{Department of Computer Science and Engineering and The Institute of Theoretical Computer Science and Communications, The Chinese University of Hong Kong. Email: {\tt syzhang@cse.cuhk.edu.hk}}}
\maketitle

\begin{abstract}
A correlated equilibrium is a fundamental solution concept in game theory that enjoys many desirable mathematical and algorithmic properties: it can achieve more fair and higher payoffs than a Nash equilibrium and it can be efficiently computed for a vast class of games. However, it requires a trusted mediator to assist the players in sampling their moves, which is a major drawback in many practical applications. 

A computational solution to this problem was proposed by Dodis, Halevi and Rabin \cite{DHR00}. They extended the original game by adding a preamble stage, where the players communicate with each other and then they perform the original game. For this extended game, they show that any correlated strategy for 2-player games can be achieved, provided that the players are computationally bounded and can communicate before the game. 

The introduction of cryptography with computational security in game theory is of great interest both from a theoretical and more importantly from a practical point of view. However, the main game-theoretic question remained open: can we achieve any correlated equilibrium for 2-player games without a trusted mediator and also unconditionally? 
 
In this paper, we provide a positive answer to this question. We show that if the players can communicate via a {\em quantum} channel before the game, then any correlated equilibrium for 2-player games can be achieved, without a trusted mediator and unconditionally.  This provides another example of a major advantage of quantum information processing: quantum communication enables players to achieve a real correlated Nash equilibrium unconditionally, a task which is impossible in the classical world. 

More precisely, we prove that for any correlated equilibrium $p$ of a strategic game $G$, there exists an extended game (with a quantum communication initial stage) $Q$ with an efficiently computable approximate Nash equilibrium $\sigma$, such that the expected payoff for both players in $\sigma$ is at least as high as in $p$.

The main cryptographic tool used in the construction is the quantum weak coin flipping protocol of Mochon \cite{Moc07}.

\end{abstract}
\thispagestyle{empty}
\clearpage
\setcounter{page}{1}
\section{Introduction}

Game theory is a research area of great importance that studies the behavior of two or more players, when interacting with each other in order to achieve individual goals. It has found far reaching applications in the fields of economics, biology, computer science, sociology, political sciences, the study of Internet and stock markets, among others.

Most games fall into two broad categories: 1) The strategic games, where all players choose their strategies simultaneously or without knowing the other players' moves. The payoffs depend on the joint strategy that is performed by the players, and the game is usually described in a matrix form {when there are only two players}; 2) The extensive games, where the players take turns in making moves. The game is described as a tree, where each node represents a stage of the game, the edges represent the possible moves and the payoffs are specified at the leaves of the tree. The players may or may not know the previous moves of the other players ({perfect} or {imperfect} information games). A strategic game is a special case of an extensive game with {imperfect} information. 

Examples of strategic games include the Battle of the Sexes, Prisoner's Dilemma, Vickrey auction, Internet routing, job scheduling, etc. Examples of extensive games include chess, the eBay auction system, evolutionary games, wars, etc. 

In order to study the optimal behavior of players in such games, the concept of an equilibrium has been put forward \cite{vNM44}. A Nash equilibrium, the most fundamental notion of an equilibrium, is a joint strategy of all players, such that no player has any incentive to change her own strategy given that all other players retain theirs. One of the seminal results in this area is that every game has a mixed Nash equilibrium \cite{vNM44,Nas51}, \ie one where the strategy of each player is a distribution over deterministic strategies. Note that these distributions are uncorrelated across different players and hence, each  player can sample independently her strategy. 

Even though the importance of Nash equilibria is unrefuted, there are some drawbacks. First, the recent breakthrough results by \cite{DGP09,CDT09} have shown that finding a Nash equilibrium is a computationally hard problem, namely it is PPAD-complete and hence it is not clear how in real life the players can decide to play according to a Nash equilibrium, when they cannot even find one in the first place. The (im)possibility of finding an efficient classical or quantum algorithm for computing Nash equilibria is one of the main open questions in the area. To make matters worse, in many games there are more than one Nash equilibrium and it is really unclear whether the players will end up in one of them, and if yes, which one and how. Note that in many cases these equilibria are not fair, and thus different players have a preference for a different equilibrium.

Let us see a simple example, the Battle of the Sexes, to illustrate the above points. A couple needs to decide where to go for holidays. Partner A prefers Amsterdam to Barcelona, and Partner B prefers Barcelona to Amsterdam. But both players prefer going to the same place than ending up in different places; see the following payoff Table, {where the pair of numbers in each entry represents the payoffs of the two partners in order}. 
\begin{center}
	\begin{tabular}{|r|c|c|}
		\hline 
		  & Amsterdam & Barcelona \\
		\hline 
		Amsterdam & (4,2) & (0,0) \\ 
		\hline
		Barcelona\hspace{.3em} & (0,0) & (2,4)	\\
		\hline	
	\end{tabular}
\end{center}
So where should they go? There are two pure Nash equilibria in the above game. They both go to Amsterdam, and hence have payoffs 4 and 2 respectively, or both go to Barcelona and have payoffs 2 and 4 respectively. Even though these are Nash equilibria, none of them is fair, causing the \emph{battle} of the sexes. 
There is actually a third Nash equilibrium, a mixed one, where each player independently flips a coin and decides to go to their preferred place with probability $2/3$ and to the preferred place of the other player with probability $1/3$. In this case, the expected payoff is the same for both players and equal to $4/3$. Even though this is a fair equilibrium, it is pretty inefficient, since now both players have payoff even less than in the case of the unfair pure equilibrium. 
Moreover, there is a $5/9$ chance the couple goes in different places, which they really do not prefer. 

One simple way to rectify all of these problems is the introduction of the notion of a correlated (Nash) equilibrium \cite{Aum74}. In such an equilibrium, we allow the strategies of the players to be drawn from a correlated distribution $p$, and same as for a Nash equilibrium, we require that each player has no incentive to deviate given the current sample of his strategy and the information of the distribution $p$ (but not the sampled strategies of the other players). There are many nice properties of these equilibria. First, they {form} a superset of Nash equilibria and hence they always exist. Moreover, {it is not hard to exhibit} games with a correlated equilibrium which enjoys fairness and whose social welfare (\ie the total payoff of the players) is arbitrarily better than that of any Nash equilibrium. Second, unlike Nash equilibria, it is computationally easy to compute an optimal correlated equilibrium by solving an LP, for many types of games, including constant-player, polymatrix, graphical, hypergraphical, congestion, local effect, scheduling, facility location, network design and symmetric games {\cite{PR08, VNRT07}}. In our previous example, a correlated equilibrium is the strategy where with probability $1/2$ the couple goes to Amsterdam or to Barcelona. The expected payoff for each player is then $3$ and the couple is equally happy. 

So, why is the notion of correlated equilibrium not the solution to all our problems? Because in general it is hard to sample from a correlated distribution. 
In fact, even for the case of two players and the distribution of the correlated equilibrium being just one fair coin, it is well known that without any computational assumptions, it is impossible to achieve just that; actually in any classical protocol one player has a strategy to get his/her desired outcome with probability 1. 
A canonical solution to this problem is to introduce a trusted mediator, who in this case flips the fair coin. However, for many real life scenarios, trusted mediators are simply not available. 

A computational solution to this problem was proposed by Dodis, Halevi and Rabin \cite{DHR00}, who showed that classical cryptographic protocols can provide an elegant way to achieve a correlated equilibrium under standard computational hardness assumptions. 
More specifically, for any strategic game where the correlated equilibrium can be efficiently computed, they do the following: Before playing the game, the players communicate in order to sample a joint strategy from the equilibrium distribution, in such a way that each player at the end of the protocol only knows her strategy and has no information about the other players' moves apart from the fact that they come from the equilibrium distribution. The privacy and correctness of this procedure {are} guaranteed by the fact that the players are computationally bounded and the assumption that a primitive, equivalent to Oblivious Transfer,  exists. Then, the players play the original game. Since they have no information about the other players' strategies and the joint strategy is sampled from a correlated equilibrium of the original game, they have no incentive to deviate. In other words, being honest during the communication phase and playing the move that resulted from the communication phase is a Nash equilibrium of the new extended game that achieves payoff equal to the correlated equilibrium of the original game. The introduction of cryptography in game theory is a very promising idea that nevertheless needs to be used with caution due to the many nuances in the two models. Note, last, that the use of cryptography by Dodis \emph{et al}. provides a solution only when one is willing to accept the notion of computational equilibria, which are very different than the equilibria used by game theorists. {Since then, a series of works have studied the relation between cryptography and game theory \cite{FS02, LMPS04, ILM05, ADGH06}.}

In our paper we show that, in fact, one need not resort to computational equilibria, if we allow the players to communicate via a quantum channel instead of a classical one. This provides another example of a major advantage of quantum information processing: quantum communication enables players to achieve a real correlated Nash equilibrium. Note that we only make the communication before the game quantum but the game itself remains a classical one. 

\emph{A priori}, it is not clear that quantum communication can provide any significant advantage, since we know that Oblivious Transfer, the primitive that Dodis, Halevi and Rabin need for their construction, is impossible even in the quantum world {\cite{Lo97}}. We overcome this problem by providing a new way to extend any game with an efficient correlated equilibrium into a new game that has an efficient Nash equilibrium achieving equal {or even better} payoffs {(up to an arbitrarily small $\varepsilon$)}. {The construction} is based on the existence of a weaker primitive, called Weak Coin Flipping. This primitive is impossible classically without any computational assumptions. In the quantum world, however, 
Mochon \cite{Moc07} has showed in a powerful result that there exists a quantum coin flipping protocol, where player A prefers Head and player B prefers Tail (which is exactly the case in the Battle of the Sexes), such that if one player plays the honest strategy, then no matter how the other player plays, the bias of the coin {cannot exceed} an arbitrarily small $\varepsilon$. 

{At a} high level, the new game we construct has the following three stages: 1) \emph{Communication stage}: 
the players use as a subroutine the quantum weak coin flipping protocol in order to sample a suggested joint strategy from the distribution of the original correlated equilibrium. Note that we do not preserve privacy of the moves, \ie at the end of this procedure both players know the joint strategy. 2) \emph{Game stage}: the players play the original game. Note of course that since each player knows the suggested strategy for the other player, it may be to her advantage to change her strategy instead of following the suggestion. We remedy this situation by using the usual ``Punishment for Deviation" method in the final stage. 3) \emph{Checking stage}: the players submit an Accept/Reject move, where a player plays Reject if the strategy of the other player during the second stage is not equal to the suggested one. The payoff of the players is equal to the one in the original game if they both play Accept in the last phase, and 0 otherwise. Note that we do not need the Accept/Reject moves to be simultaneous.

It is not hard to see that being honest during the communication stage and playing the suggested move is an approximate Nash equilibrium for this game and it achieves payoff equal to the correlated equilibrium of the original game. Let us assume that one of the players is dishonest while the other is playing the honest strategy. {The cheating player} can deviate during the coin flipping process but this will only increase his payoff by at most an $\varepsilon$ fraction by the security of the coin flipping protocol. Then, in the second stage, he can deviate by not playing the suggested strategy, but then his payoff will be 0 since the honest player will play Reject in the Checking stage. Hence, there is no significant advantage for any player to deviate from the honest strategy.

\begin{Thm}
For any correlated equilibrium $p$ of a game $G$, there exists an extended game $Q$ with a Nash equilibrium $\sigma$, such that the expected payoff for both players in $\sigma$ is at least as high as in $p$.
\end{Thm}

Let us make a more detailed comparison with the results of Dodis, Halevi and Rabin \cite{DHR00}. They describe an extended game, {first introduced by Barany \cite{Bar92}}, that involves a communication stage and then the game stage. In the communication {stage}, they securely compute a functionality that they call Correlated Element Selection. This consists of two players sampling a joint strategy from a correlated distribution, with the extra privacy property that at the end each player knows only his/her own move. Their construction is based on a Blindable Encryption scheme, which is as strong as Oblivious Transfer, meaning that it is complete for all secure multiparty computation. Then, {in the second stage,} the players play the original game. 

In our protocol, the communication stage achieves something weaker. We just sample from the correlated distribution in a way that at the end, both players know the joint strategy. By removing the privacy constraint we are able to achieve the sampling using the weaker primitive of Weak Coin Flipping. In both protocols we need to dissuade the players from cheating. In the DHR protocol, if a player catches the other one cheating during the communication stage, then he plays his minmax move in the second stage, hence minimizing the other players payoff. Moreover, no one has an incentive to play a different move than the suggested one (except with exponentially small probability), hence they do not have to worry about forcing the players to play the suggested move.  We are more explicit in our punishment by adding the Accept/Reject stage, both for dissuading the players from cheating during the communication stage, but more importantly in order to dissuade the players from playing a different move from the suggested one. Last, both protocols achieve an approximate equilibrium, since the communication part is not perfect.

On the other hand, we achieve something much stronger than before, which is that we do not make any assumptions about the computational power of the players. Hence we are able to use quantum communication to achieve a real correlated equilibrium for a large array of different types of games with unconditionally powerful players and without a trusted mediator. 

A few remarks are in order for this extra checking stage that we add to the original game.
First, note that all the equilibria remain unchanged, since we specified the payoffs of any joint strategy with a Reject move as 0. Hence, sampling a correlated equilibrium in the new game is equivalent to sampling a correlated equilibrium in the original game. This means that the quantum advantage comes from the sampling part and not due to the checking part. For a fair comparison, we can also augment the classical game with the choice of Accept/Reject. It is not hard to see that the players still cannot sample a correlated equilibrium in this new game. In particular, consider the bimatrix game, in which the only fair correlated equilibrium with total payoff 1 is to choose $(0,0)$ and $(1,1)$ each with half probability. This is a (fair) coin flip, which is known to be impossible to sample. Even if we extend the game with the Accept/Reject stage, the only fair correlated equilibrium with total payoff 1 is to choose $(0,\text{Accept};0,\text{Accept})$ and $(1,\text{Accept};1,\text{Accept})$ each with half probability. But this results in flipping a coin again. Note that we use the normal-form for convenience even though the Accept/Reject stage is not simultaneous.
\begin{center}
	\begin{tabular}{|r|c|c|}
		\hline 
		  & 0 & 1 \\
		\hline 
		0 & (1,0) & (0,0) \\ 
		\hline
		1 & (0,0) & (0,1)	\\
		\hline	
	\end{tabular}$\quad\quad$
\begin{tabular}{|r|c|c|c|c|}
		\hline 
		  & (0,Accept) & (1,Accept) & (0,Reject) & (1,Reject) \\
		\hline 
		(0,Accept) & (1,0) & (0,0) & (0,0) & (0,0) \\ 
		\hline
		(1,Accept) & (0,0) & (0,1) & (0,0) & (0,0) \\
		\hline	
		(0,Reject) & (0,0) & (0,0) & (0,0) & (0,0) \\ 
		\hline
		(1,Reject) & (0,0) & (0,0) & (0,0) & (0,0) \\
		\hline	
	\end{tabular}
\end{center}

{Second, in many cases, the checking stage is actually not necessary. For example, for a large class of games, an optimal correlated equilibrium has support only on a set of pure Nash equilibria.  
In these correlated equilibria, knowing the other players' sampled strategies $s_{-i}$ does not give any incentive for Player $i$ to deviate from her own sampled strategy $s_i$, since $s$ is already a pure Nash equilibrium. Note that our previous example of Battle of the Sexes falls into this family of games.
}

Third, in many practical situations, breaking preagreed rules is considered losing (and thus given the least payoff) automatically. Many games in sports are of this nature. For example, when the referee tosses a coin to decide the side of the court for each team, both teams know the outcome of this random process and are not allowed to disagree no matter the outcome; otherwise the team will be claimed to lose by the referee immediately. Moreover, in extensive games, the checking phase is already implicitly present. In the middle of a chess game, only a subset of moves is compatible with the stage of the game and hence if a player decides to play some other move, then the other player will Reject either immediately or at the end of the game. Hence, adding an Accept/Reject stage only makes explicit what is implicitly present in any game, that if a player breaks the rules then the other one rejects the outcome of the game.  

Fourth, our Accept/Reject stage is not simultaneous. One has to be very careful with adding simultaneous moves to a game, since two players can flip a fair coin with a simultaneous move where each plays one of two possible moves at random. If the two moves are the same then the coin is Head and if different the coin is Tail. Here, we do not add the ability to play simultaneously. 

Going back to our Battle of the Sexes game, from the properties of the quantum coin flip, it is not hard to see that the strategy of both players being honest during the coin flip is a $2\varepsilon$-approximate Nash equilibrium with payoff $3$ for both players (equal to the original correlated Nash equilibrium); if one player decides to play any other strategy while the other one remains honest, then {the cheating player's} payoff will be {no more than} than $(1/2+\varepsilon)\cdot 4 + (1/2-\varepsilon) \cdot 2 =  3+ 2\varepsilon$ (In later sections, we normalize the game by scaling all utilities to be within $[0,1]$ for a fair comparison). As we have said, one can easily generalize this and other games so that the correlated Nash equilibria be made arbitrarily better and more fair than all the Nash equilibria, which are the only ones that can be sampled classically without a trusted mediator. 

Note that our protocol does not provide a quantum algorithm to compute a Nash equilibrium. However, it almost renders this question moot. Instead of a quantum algorithm to compute a Nash equilibrium, there is a quantum protocol where the players can generate a correlated equilibrium, {which enjoys desirable properties such as fairness and higher payoff}. 

{Since our protocol uses quantum channels, one may wonder whether the power of two-way quantum communication enables us to achieve any \emph{quantum} equilibrium with payoff higher than any classical correlated equilibrium. This is actually not possible: Any quantum protocol eventually generates a joint strategy $s$ according to some correlated distribution $p$. If the players' behaviors in the protocol form a Nash equilibrium (in the sense that no player has any incentive to use other sequence of quantum operations), then the resulting distribution $p$ is a quantum correlated equilibrium of the quantized game, because otherwise the players would like to change their behaviors in the last step. By an observation in \cite{Zha10}, $p$ is also a (classical) correlated equilibrium of the original (classical) game, which the present paper already gives a way to generate.}

\section{Preliminaries}

\subsection{Game Theory}
In a classical strategic game {with $n$ players, labeled by $\{1,2,\ldots,n\}$}, each player $i$ has a set $S_i$ of strategies. 
We use $s=(s_1,\ldots, s_n)$ to denote the \emph{joint strategy} selected by the players and $S= S_1 \times \ldots \times S_n$ to denote the set of all possible joint strategies. Each player $i$ has a utility function $u_i: S \rightarrow \mbR$, specifying the \emph{payoff} or \emph{utility} $u_i(s)$ to player $i$ on the joint strategy $s$. For simplicity of notation, we use subscript $-i$ to denote the set $[n]-\{i\}$, so $s_{-i}$ is $(s_1, \ldots, s_{i-1}, s_{i+1}, \ldots, s_n)$, and similarly for $S_{-i}$, $p_{-i}$, etc. 

In a classical extensive game with perfect information, the players take moves in turns and all players know the entire history of all players' moves. An extensive game can be transformed into strategic form by tabulating all deterministic strategies of the players, which usually results in an exponential increase in size.

A game is \emph{$[0,1]$-normalized}, or simply \emph{normalized}, if all utility functions are in $[0,1]$. Any game can be scaled to a normalized one. For a fair comparison, we assume that all games in this paper are normalized.

A Nash equilibrium is a fundamental solution concept in game theory. Roughly, it says that in a joint strategy, no player can gain more by changing her strategy, provided that all other players keep their current strategies unchanged. 
\begin{Def}
A \emph{pure Nash equilibrium} is a joint strategy $s = (s_i, \ldots ,s_n) \in S$ satisfying 
\begin{align*}
	u_i(s_i,s_{-i}) \geq  u_i(s_i',s_{-i}), \qquad \forall i\in [n], \forall s'_i\in S_i.
\end{align*}
\end{Def}
Pure Nash equilibria can be generalized by allowing each player to independently select her strategy according to some distribution, leading to the following concept of \emph{mixed Nash equilibrium}. 
\begin{Def}
A \emph{(mixed) Nash equilibrium (NE)} is a product probability distribution $p = p_1 \times \ldots \times p_n$, where each $p_i$ is a probability distributions over $S_i$, satisfying 
\begin{align*}
	\sum_{s_{-i}} p_{-i}(s_{-i}) u_i(s_i,s_{-i}) \geq  \sum_{s_{-i}} p_{-i}(s_{-i}) u_i(s_i',s_{-i}), \quad \forall i\in [n], \ \forall s_i, s'_i\in S_i \text{ with } p_i(s_i)>0.
\end{align*}
\end{Def}
A Correlated equilibrium assumes an external party to draw a set of strategies for the players according to a probability distribution, possibly correlated in an arbitrary way, over $S$, and suggest them to each player. If player $i$ receives a suggested strategy $s_i$, the player can never increase its expected utility by switching to another strategy $s_i' \in S_i$, assuming that all other players are all going to choose their received suggestion $s'$. 

\begin{Def} \label{thm:CE}
A \emph{correlated Nash equilibrium (CE)} is a probability distribution $p$ over $S$ satisfying
\begin{align*}
	\sum_{s_{-i}} p(s_i,s_{-i}) u_i(s_i,s_{-i}) \geq  \sum_{s_{-i}} p(s_i,s_{-i}) u_i(s_i',s_{-i}), \qquad \forall i\in [n], \ \forall s_i, s'_i\in S_i.
\end{align*}
\end{Def}

{We will also need an approximate version of equilibrium, which basically says that no Player $i$ can gain much by changing the suggested strategy $s_i$. Depending on whether we require the limit of the gain for each possible $s_i$ in the support of $p$ or on average of $p$, one can define worst-case and average-case approximate equilibrium. It turns out that the average-case one, as defined below, has many nice properties, such as being the limit of a natural dynamics of minimum regrets (\cite{VNRT07}, Chapter 4) {and hence it is the one we will use}.
\begin{Def}\label{def: approx CE}
An \emph{$\varepsilon$-correlated equilibrium} is a probability distribution $p$ over $S$ satisfying
\begin{align*}
	\av_{s\leftarrow p}[u_i(s_i'(s_i),s_{-i})] \leq \av_{s\leftarrow p}[u_i(s_i,s_{-i})] + \varepsilon,
\end{align*}
for any $i$ and any function $s_i': S_i \rightarrow S_i$. An $\varepsilon$-correlated equilibrium $p$ is an \emph{$\varepsilon$-Nash equilibrium} if it is a product distribution $p = p_1\times \cdots \times p_n$. 
\end{Def}


We can also define equilibria for extensive games by defining the corresponding equilibria on their strategic form.

\subsection{Cryptography}

We provide the formal definition of a weak coin flipping protocol.

\begin{Def}
A {\em weak coin flipping} protocol between two parties Alice and Bob is a protocol where Alice and Bob interact and at the end, Alice outputs a value $c_A \in \{0,1\}$ and Bob outputs a value $c_B \in \{0,1\}$. If $c_A = c_B$, we say that the protocol outputs $c = c_A$. If $c_A \neq c_B$ then the protocol outputs $c = \bot$. 

{An $(a,\varepsilon)$-}weak coin flipping protocol  ($WCF(a,\varepsilon)$) has the following properties:
\begin{itemize}
\item If $c = a$, we say that Alice wins. If $c = 1-a$, we say that Bob wins.
\item If Alice and Bob are honest then $\pr[\mbox{Alice wins}] = \pr[ \mbox{Bob wins}] = 1/2$
\item If Alice cheats and Bob is honest then $P^*_A = \pr[ \mbox{Alice wins}] \le 1/2 + \varepsilon$
\item If Bob cheats and Alice is honest then $P^*_B = \pr[ \mbox{Bob wins}]\le 1/2 + \varepsilon$
\end{itemize}
 
$P^*_A$ and $P^*_B$ are the cheating probabilities of Alice and Bob. The cheating probability of the protocol is defined as $\max\{P^*_A,P^*_B\}$. 
\end{Def}

Note that in the definition the players do not abort, since a player that wants to abort can always declare victory rather than aborting without reducing the security of the protocol.

We will use the following result by Mochon.
\begin{Prop}{\em\cite{Moc07}}\label{Mochon}
For every $\varepsilon > 0$ and $a \in \{0,1\}$, there exists a quantum $WCF(a,\varepsilon)$ protocol $P$. 
\end{Prop}

Note that this is a weaker definition of a usual coin flip, since here, we assign a winning value for each player. Even though each player cannot bias the coin towards this winning value, he or she can bias the coin towards the losing value with probability 1. Weak coin flipping is possible using quantum communication, though for the strong coin flipping the optimal cheating probability for any protocol is $1/\sqrt{2}$ {\cite{Kit03,CK09}}.

In the following section we will use weak coin flipping as a subroutine for the following cryptographic primitive, that enables two players to jointly sample from a correlated distribution, in a way that no dishonest player can force a distribution which is far from the honest one.

\begin{Def}
A {\em Correlated Strategy Sampling} protocol between two players $P_1$ and $P_2$ is an interactive protocol where the players receive as input a game $G$ with an efficiently computable correlated equilibrium $p$ and at the end, $P_1$ outputs a joint strategy $(s_1,s_2) \in S_1 \times S_2$ and $P_2$ outputs a joint strategy $(s_1',s_2') \in S_1 \times S_2$. If $(s_1,s_2) = (s_1',s_2')$, we say the protocol outputs $s=(s_1,s_2)$. If $(s_1,s_2) \neq (s_1',s_2')$ then we say the protocol outputs $s = \bot$. 

An $(\varepsilon,\delta)$-Correlated Strategy Sampling procedure satisfies the following properties:
\begin{enumerate}
\item If both players follow the honest strategy, then they both output the same joint strategy $s=(s_1,s_2)$, where $(s_1,s_2) \leftarrow p_h$ {for some distribution $p_h$}, \st
\[ \mbox{for {both} } i \in \{1,2\}, \quad \av_{(s_1,s_2)\leftarrow p_h} [u_i(s_1,s_2)] \geq \av_{(s_1,s_2)\leftarrow p} [u_i(s_1,s_2)] {- \delta}
\]
\item If Player $1$ is dishonest and Player 2 is honest (similarly for the other case), then Player $2$ outputs a joint strategy $(s_1,s_2)$ {distributed according to some $q$}, \st 
	\[\av_{(s_1,s_2)\leftarrow q} [u_2(s_1,s_2)] \geq \av_{(s_1,s_2)\leftarrow p_h} [u_2(s_1,s_2)] - \varepsilon\] 
\[\av_{(s_1,s_2)\leftarrow q} [u_1(s_1,s_2)] \leq \av_{(s_1,s_2)\leftarrow p_h} [u_1(s_1,s_2)] + \varepsilon\] 
\end{enumerate}
\end{Def}

Note again, that similar to the case of the weak coin flip, the players do not abort, since a player that wants to abort can always choose the joint strategy that is best for him rather than aborting without reducing the security of the protocol.


\section{The extended game}


For simplicity, we consider a two-player strategic game $G$ of size $n$, but our results easily extend to more players. We describe how to derive an extended game $Q$ from any such $G$.

Similar to the DHR extended game, we assume that the players can communicate with each other before they start playing the game, but now via a quantum channel. In this preamble stage they perform a quantum protocol that we call Correlated Strategy Sampling. 

In the following section we show how to implement this procedure unconditionally, using a Weak Coin Flipping subroutine with bias $\varepsilon'{=O(\varepsilon/\log n)}$.

Then, we extend the original game G to a 2-stage game, where the first stage is identical to the game $G$ and for the second stage, which we call the Checking stage, the available moves for each player are Accept or Reject. We define the payoff for any joint strategy where some player outputs Reject in the second stage to be 0. 

\COMMENT{
\begin{center}
	\begin{tabular}{|r|c|c|}
		\hline 
		  & ($S_2$, Accept) & ($S_2$, Reject) \\
		\hline 
		($S_1$, Accept) & M & [(0,0)] \\ 
		\hline
		($S_1$, Reject) & [(0,0)] & [(0,0)]	\\
		\hline	
	\end{tabular}
\end{center}
where $[(0,0)]$ is the bimatrix of all entries being $(0,0)$.
}
\COMMENT{
\begin{figure*}[h]\label{Fig}
\centering
\includegraphics[width=1.0\textwidth]{figure1.pdf}
\caption{Extended Game}
\end{figure*}
}



\begin{center}
\fbox{
\begin{minipage}[l1pt]{6.0in}
{\bf  Extended Game $Q$} 

\begin{enumerate}
	\item 
	\textbf{Communication Stage}: The two players perform the Correlated Strategy Sampling procedure for the game $G$ and correlated equilibrium $p$. 
\item 
\textbf{Game Stage}: The two players play the original game $G$.
\item
\textbf{Checking Stage}: The two players {each} play a move from the set $\{A, R\}$.
\end{enumerate}
\end{minipage}
}
\end{center}
We can now restate and prove our main theorem.

\setcounter{Thm}{0}
\begin{Thm}
For any correlated equilibrium $p$ of the game $G$ of size $n$, and for any $\varepsilon, \delta >0$, there exists an extended game $Q$ with an $\varepsilon$-Nash equilibrium $\sigma$ that can be computed in time $poly(n,1/\delta, 1/\varepsilon)$ and such that the {expected} payoff for both players in $\sigma$ is at least as high as the one in $p$ minus $\delta$. 
\end{Thm}

\begin{proof}
We describe Player $1$'s strategy in the $\varepsilon$-Nash equilibrium $\sigma$ as follows (Player 2's strategy is symmetric): In the Communication Stage, Player 1 is honest and {obtains an output} $(s_1,s_2)$. In the Game Stage, he plays the move $s_1$. In the Checking Stage, he plays $A$ if Player 2's move {in the Game Stage was $s_2$} and $R$ otherwise.

Let us show that this is indeed an {$\varepsilon$-}Nash equilibrium. A dishonest player (assume Player 1) can try to increase his payoff by first deviating from the protocol in the Communication Stage. If Player 2 outputs a joint strategy $(s_1,s_2)$ then we know from the security of the Correlated Strategy Sampling procedure that this is a sample from a distribution $q$ \st
\[\av_{(s_1,s_2)\leftarrow q} [u_1(s_1,s_2)] \leq \av_{(s_1,s_2)\leftarrow p_h} [u_1(s_1,s_2)] + \varepsilon\] 
Hence, {if Player 1 is dishonest during Stage 1 and then plays $s_1$ in Stage 2, then his gain is} at most $\varepsilon$. If he decides to change his move, then the honest player would play $R$ in Stage 3, so his payoff would be 0. Overall, no matter what strategy the dishonest player follows he cannot increase his payoff more than $\varepsilon$ from the honest strategy mentioned above, and hence this strategy is an $\varepsilon$-approximate Nash equilibrium. 
\end{proof}

Note that from the security of the Correlated Strategy Sampling procedure we also have 
\[\av_{(s_1,s_2)\leftarrow q} [u_2(s_1,s_2)] \geq \av_{(s_1,s_2)\leftarrow p_h} [u_2(s_1,s_2)] - \varepsilon\]
Hence, we have the following interesting corollary
\begin{Cor}
In the extended game $Q$, the expected payoff of the honest player will not decrease by more than $\varepsilon$, no matter how the dishonest player deviates, unless the dishonest player makes both players' payoff equal to 0.
\end{Cor}
In other words, the honest strategy remains an equilibrium even if the objective of a player is not to maximize his own payoff but rather maximize the difference between the players' payoffs.

\section{The Correlated Strategy Sampling procedure}

Let us start by fixing some notation. {In a two-player game, let $p$ be an efficiently computable correlated equilibrium that the players know and aim to generate. A typical scenario is that $p$ is the lexicographically first correlated equilibrium that maximizes the total payoff.} Let $p$ be the distribution of the CE of the original game $G$ of size $n$. If the distribution is not uniform we can emulate it by a uniform distribution on a multiset of size $K={n/\delta}$ (we choose $K=2^k$) and the distance between the two distributions is an inverse polynomial of $n$.  Let $p_h$ be the distribution that arises when both players are honest and $q$ the distribution of the honest player's output when the other player is dishonest. All distributions are on $\{0,1\}^k$.

Let us also define the following distributions for all $m \in {\{0,1,...,k\}}$:
\begin{eqnarray*}
p_h^{c^1,\ldots,c^m}:\{ 0,1\}^{k-m}\rightarrow {\mbR} & \mbox{ s.t. } & p_h^{c^1,\ldots,c^m}(r^{m+1},\ldots, r^k) = \frac{p_h(c^1,\ldots,c^m,r^{m+1},\ldots,r^k)}{\sum_{r^{m+1},\ldots, r^k}p_h(c^1,\ldots,c^m,r^{m+1},\ldots,r^k)}\\
q^m:\{0,1\}^m\rightarrow {\mbR} & \mbox{ s.t. } & q^m(c^1,\ldots,c^m)= \sum_{r^{m+1},\ldots,r^k}  q(c^1,\ldots,c^m,r^{m+1},\ldots,r^k)\\
p_h^m:\{0,1\}^m\rightarrow {\mbR} & \mbox{ s.t. } & p_h^m(c^1,\ldots,c^m)= \sum_{r^{m+1},\ldots,r^k}  p_h(c^1,\ldots,c^m,r^{m+1},\ldots,r^k)
\end{eqnarray*}
Note that $q^0 = p_h^0 = 1$. Also, by $(s_1,s_2)\leftarrow p_h^{c^1,\ldots,c^m}$ we mean the distribution on joint strategies $(s_1,s_2)$ or equivalently $\ell$-bit strings that arises from $p_h$ conditionned on the first $m$ bits of $\ell$ being $c^1,\ldots,c^m$.
Let ${sign}(a)$ be the function which is 1 if $a\geq 0$ and $-1$ if $a<0$. 
The protocol appears in the following figure.

\begin{center}
\fbox{
\begin{minipage}[l1pt]{6.0in}
{\bf  $(\varepsilon,\delta)$-Correlated Strategy Sampling Protocol} \\

{\bf Input}: A game $G$ of size $n$ with an efficiently computable correlated equilibrium $p$.
\begin{enumerate}
	\item Each Player $i$ computes locally the equilibrium $p$ and emulates $p$ by a uniform distribution on a multiset of joint strategies $\{(s_1^\ell,s_2^\ell)\}_{\ell \in \{ 0,1\}^k}$, with $k=O(\log n)$.

	\item   \textbf{ for} $j = 1$ to $k$
	\begin{enumerate}
		\item Each Player $i$ computes and announces his preference \\
$a_i^j = 
sign \big(E_{(s_1,s_2)\leftarrow p_h^{c^1,\ldots,c^{j-1},0}}[u_i(s_1,s_2)] - E_{(s_1,s_2)\leftarrow p_h^{c^1,\ldots,c^{j-1},1}}[u_i(s_1,s_2)] \big)$. 

		\item \textbf{if} $a^j_1a^j_2 = -1$, 
			
		\quad Run $WCF(a^j_1,\varepsilon /{(2k)})$. Let the outcome of Player $i$ be $c_i^j \in \{0,1\}$.
		
		\textbf{else}
		
		\quad Set $c_1^j=c_2^j=a^j_1$, their commonly desirable value.		
	\end{enumerate}
	\item Each Player $i$ outputs $(s_1^\ell,s_2^\ell)$, with $\ell = c_i^1...c_i^{k}$
\end{enumerate}
\end{minipage}
}
\end{center}

\paragraph{Analysis}
First, if both players are honest then their expected utility is at least as high as in the original CE, up to an additive error {$\delta$} due to the precision of using $k$ bits to emulate $p$. 
If in all rounds they flip a fair coin then their expected utility is exactly the same as in $p$. If at some round they both agree on a preferred value then this increases both players expected utility.

We now prove that no dishonest player can increase his utility by much compared to his honest utility. Without loss of generality, we assume that a dishonest player will always announce his preference honestly and then try to win the weak coin flipping protocol (if their preferences differ), since this can only increase his utility.

Let us assume {without loss of generality} that Player 1 is dishonest {and Player 2 is honest}. We will prove that after round $m$, we have
\begin{Claim}\label{induction}
$||q^m - p_h^m||_1 \leq m \frac{\varepsilon}{k}.$
\end{Claim}
The proof is in the next subsection. By the Claim, after the $k$-th round we have
\[ ||q^k-p_h^k||_1 = ||q-p_h||_1 \leq \varepsilon
\]
This means, that the expected utility of the dishonest player over the distribution $q$ is 
\[\av_{(s_1,s_2)\leftarrow q} [u_1(s_1,s_2)] \leq \av_{(s_1,s_2)\leftarrow p_h} [u_1(s_1,s_2)] + \varepsilon\] 
since all the utilities are normalized. Moreover, for the honest player we have
\[\av_{(s_1,s_2)\leftarrow q} [u_2(s_1,s_2)] \geq \av_{(s_1,s_2)\leftarrow p_h} [u_2(s_1,s_2)] - \varepsilon\] 

The same analysis holds when Player 2 is dishonest. Also, {it is easy to see that} the complexity of the protocol is polynomial in $n{/\delta}$ and $1/\varepsilon$. {This completes the proof of our main theorem.}

\medskip A final remark is that the same protocol can be used for general $k$-player games. In each round, some players prefer $c^m$ to be 0 and some players prefer 1. We can then let two representatives, one from each group, to do the weak coin flipping, at the end of which the representatives announce the bits. If one representative lies, then the other reject in the third stage. The previous analysis then easily applies to this scenario as well.

\subsection{Proof of Claim \ref{induction}}

{The base case of $m=0$ is trivial because $q^0 = p_h^0 = 1$}.

{Let us assume now that after round $m$, we have proved the above inequality, then for the $(m+1)-th$ round, there are two cases. First, the players agree on the bit, then the distance of the distributions $\|q^{m+1} - p_h^{m+1}\|_1$ remains the same as $\|q^m - p_h^m\|_1$. The second case is that the two players have different preferences for $c^{m+1}$. Assume without loss of generality that Player 1 prefers 0 and Player 2 prefers 1. We have}
\begin{eqnarray*}
& &\lefteqn{||q^{m+1}-p^{m+1}||_1}\\
 & = & \sum_{c^1,\ldots,c^{m+1}}\Big|  q^{m+1}(c^1,\ldots,c^{m+1}) - p_h^{m+1}(c^1,\ldots,c^{m+1}) \Big| \\
& = & \sum_{c^1,\ldots,c^{m+1}}  \Big| \sum_{r^{m+2},\ldots,r^k}  
\big[ q(c^1,\ldots,c^{m+1},r^{m+2},\ldots,r^k) - p_h(c^1,\ldots,c^{m+1},r^{m+2},\ldots,r^k)\big]  \Big|  \\
& = & \sum_{c^1,\ldots, c^m}  \Big( \big| \sum_{r^{m+2},\ldots,r^k}  
\big[ q(c^1,\ldots,c^{m},0,r^{m+2},\ldots,r^k) - p_h(c^1,\ldots,c^{m},0,r^{m+2},\ldots,r^k)\big]  \big| \\
& + &  \big| \sum_{r^{m+2},\ldots,r^k}  
\big[ q(c^1,\ldots,c^{m},1,r^{m+2},\ldots,r^k) - p_h(c^1,\ldots,c^{m},1,r^{m+2},\ldots,r^k)\big]  \big| \Big)\\
& \leq & \sum_{c^1,\ldots, c^m} \Big( \big|   \sum_{r^{m+1},\ldots,r^k}  
\big[ (\frac{1}{2}+\frac{\varepsilon}{{2}k}) q(c^1,\ldots,c^{m},r^{m+1},r^{m+2},\ldots,r^k) - \frac{1}{2} p_h(c^1,\ldots,c^{m},r^{m+1},r^{m+2},\ldots,r^k)\big]  \big|  \\
& + & \big|   \sum_{r^{m+1},\ldots,r^k}  
\big[ (\frac{1}{2}-\frac{\varepsilon}{{2}k}) q(c^1,\ldots,c^{m},r^{m+1},r^{m+2},\ldots,r^k) - \frac{1}{2} p_h(c^1,\ldots,c^{m},r^{m+1},r^{m+2},\ldots,r^k)\big]  \big| \Big) \\
& \leq & ||q^m-p_h^m||_1 + \frac{\varepsilon}{k}\cdot 1 \\
& \leq & (m+1)\frac{\varepsilon}{k}
\end{eqnarray*}

In the above inequalities we used the fact that the bias of the weak coin flipping protocol is always less than $\varepsilon/2k$ no matter what the values of the previous coins are. In fact, the security proof for the bias holds against a dishonest player that may possess a quantum auxiliary input, hence includes the situation where the dishonest player may try to entangle the different executions of the coin flips.

\subsection*{Acknowledgments}
Most of the work was done when the authors visited Centre of Quantum Technologies (CQT), Singapore in early January, 2011, under the support of CQT.

I.K.'s research was also supported by French projects ANR-09-JCJC-0067-01, ANR-08-EMER-012 and the project QCS (grant 255961) of the E.U. \hspace{1em} S.Z.'s research was supported by China Basic Research Grant 2011CBA00300 (sub-project 2011CBA00301), Hong Kong General Research Fund 419309 and 418710, and benefited from research trips under the support of China Basic Research Grant 2007CB807900 (sub-project 2007CB807901).

\bibliography{CE-QCF}
\bibliographystyle{plain}

\end{document}